\documentclass[12pt]{article}
\usepackage{amssymb,amsmath}
\usepackage{graphicx}
\usepackage{amsfonts}
\usepackage{float}
\usepackage{flafter}
\usepackage{color}

\labelwidth\leftmargini\advance\labelwidth-\labelsep

\def\R{\relax\ifmmode I\!\!R\else$I\!\!R$\fi}

\def\Z{\relax\ifmmode Z\!\!\!Z\else$Z\!\!\!Z$\fi}

\def\C{\relax\ifmmode C\!\!\!\!I\else$C\!\!\!\!I$\fi}

\def\K{\relax\ifmmode I\!\!K\else$I\!\!K$\fi}

\def\N{\relax\ifmmode I\!\!N\else$I\!\!N$\fi}

\newcounter{defcounter}[section]
{\vspace{0.1cm}\begin{sloppypar}\noindent\stepcounter{defcounter}{\bfseries
Definition
      \thesection.\thedefcounter}}%
{\end{sloppypar}\vspace{0.1cm}}
\newtheorem{conjecture}{Conjecture}[section]

\newtheorem{proposition}{Proposition}[section]

\newcommand{\proof}{{\bf Proof.} }

\newcommand{\qed}{\hfill $\square$}

\begin{document}
\thispagestyle{empty}
\begin{center}
{\Large {\bf Recursive overlap Bernoulli distributions and an entropy concavity conjecture}}
\end{center}
\begin{center}J\"org Neunh\"auserer\\
Technical University of Braunschweig \\
joerg.neunhaeuserer@web.de
\end{center}
\begin{center}
\begin{abstract}
We introduce a family of recursively generated finite probability distributions obtained from left and right embeddings with overlaps. The construction interpolates between the classical binomial distribution and the non-overlapping Bernoulli product distribution. We derive explicit formulas for the expectation, variance, and the generating function of higher moments, and formulate a conjecture asserting that the Shannon entropy is concave. The conjecture is proved in the two extremal cases and supported by symbolic computations for numerous overlap sequences.
\end{abstract}
\end{center}
\section{Introduction}
In this note we introduce a family of recursive overlap Bernoulli distributions generated in the following way.   
Let $(a_{n})$ be a strictly increasing sequence of integers with $a_{1}=2$ and let $p\in[0,1]$. Define probability vectors $P_n(p)\in\mathbb R^{a_n}$ by
$
P_1(p)=(p,1-p)$
and $$
P_{n+1}(p)
=
p\,L(P_n(p))+(1-p)\,R(P_n(p)),
$$
where the “left’’ and “right’’ embedding operators
$
L:\mathbb R^{a_n}\to\mathbb R^{a_{n+1}},
R:\mathbb R^{a_n}\to\mathbb R^{a_{n+1}}
$
are given by
$$
L(P_n(p))=(P_n(p),0,\ldots,0),
\quad
R(P_n(p))=(0,\ldots,0,P_n(p)).
$$ 
Let $X_{n}(p)$ be the corresponding random variable taking values in $\{1,\dots, a_{n}\}$. If $(a_{n})=(n+1)$ the random variables $X_{n}(p)$ are binomial distributed with
\[ \mathbb{P}(X_{n}(p)=k)=\binom{n}{k-1}p^{n-k+1} (1-p)^{k-1}\] 
for $k=1,\dots,n+1$. If $(a_{n})=(2^n)$ the random variables form the non-overlapping Bernoulli with  
\[ \mathbb{P}(X_{n}(p)=k)=p^{n-\sharp_{n}(k)} (1-p)^{\sharp_{n}(k)}\] 
for $k=1,\dots,2^{n}$, where $\sharp_{n}(k)$ is the number of digits $1$ in the dyadic representation of $k-1$. If $a_{n+1}\ge 2a_n$ no overlap occurs and $X_n(p)$ has the same distribution with values in a subset of $\{1,\dots,a_n\}$ of cardinality $2^n$. If $n+1<a_{n+1}<2a_n$ for some $n$, then the two copies overlap. The overlap creates new probability distributions which interpolate between the binomial distribution and the separated Bernoulli convolutions.\\
In the next section we will calculate expectation, variance and moment generating function of the random variables $X_{n}(p)$. In section 3 we state the entropy conjecture 
that the Shannon entropy of $X_{n}(p)$ is concave in $p$ for all $n\ge 1$ and hence an unimodal function with maximum at $p=1/2$. We prove this conjecture in the binomial case $a_{n}=n+1$ and for all sequences with $a_{n+1}\ge 2a_{n}$. In the last section we present computational evidence that the conjecture holds in general.    
The motivation for this construction and the entropy conjecture is twofold.
First, the construction generalizes random walks on infinite rooted self-similar
graphs studied in \cite{NEU}. For example, for the recursive sequence
\[
a_{n+1}=a_n+a_{n-1}+1,\qquad a_0=1,\ a_1=2,
\]
we recover the directed random walk on the Fibonacci graph.
Second, the construction is closely related to 
infinite convolved Bernoulli measures, which play an important role in fractal geometry; see Chapter~8 of
\cite{BSS} and the references therein.
If Conjecture~3.1 is true, then the dimension function of the corresponding
singular Bernoulli measures would be concave.
\section{Expectation, variance and higher moments}
We first give here the expectation and variance of the random variables defined in the introduction. 
\begin{proposition}
For a strictly increasing sequence of integers $(a_{n})$ with $a_{1}=2$ and $p\in[0,1]$ the expectation of $X_{n}(p)$ is
\[ \mathbb{E}(X_{n}(p))=p+a_{n}(1-p) \]
the variance is 
\[ \mathbb{V}(X_{n}(p))=p(1-p)(1+\sum_{k=1}^{n-1}(a_{k+1}-a_{k})^2) \]
\end{proposition}
\proof From the definition we have
\[ 
X_{n+1}(p) = \begin{cases} 
      X_{n} & \text{with probability p } \\
      X_{n}+(a_{n+1}-a_{n})&\text{with probability (1-p) }       
   \end{cases}
\]
and $X_{1}\in\{1,2\}$ with probabilities $p$ and $(1-p)$.  Hence 
\[ \mathbb{E}(X_{n+1}(p))= \mathbb{E}(X_{n}(p))+(1-p)(a_{n+1}-a_{n})\]
with $\mathbb{E}(X_{1})=p+2(1-p)$. The formula for the expectation follows since the sum is telescoping. Moreover we get
\[ \mathbb{V}(X_{n+1}(p))= \mathbb{V}(X_{n}(p))+p(1-p)(a_{n+1}-a_{n})^2\]
with $\mathbb{V}(X_{1})=p(1-p)$. The formula for variance follows by induction. \qed\\~\\ 
Now we consider the generating functions for higher moments. 
\begin{proposition}
For a strictly increasing sequence of integers $(a_{n})$ with $a_{1}=2$ and $p\in[0,1]$ the moment generating function of $X_{n}(p)$ is given by
\[
M_n(t)=\mathbb{E}\bigl[e^{tX_n}\bigr]
=
e^t\bigl(p+(1-p)e^t\bigr)\prod_{k=1}^{n-1}\Bigl(p+(1-p)\,e^{\,t\,(a_{k+1}-a_k)}\Bigr).
\] 
\end{proposition}
\proof
We have
\[ 
M_{n+1}(t)
=\mathbb{E}\bigl[e^{tX_{n+1}}\bigr]
=p\,\mathbb{E}\bigl[e^{tX_n}\bigr]
+(1-p)\,\mathbb{E}\bigl[e^{t(X_n+a_{n+1}-a_{n})}\bigr].
\]
\[
=p\,M_n(t)+(1-p)e^{t(a_{n+1}-a_{n})}M_n(t)
=\bigl(p+(1-p)e^{t(a_{n+1}-a_{n})}\bigr)\,M_n(t).
\]
Since \[
M_1(t)
=\mathbb{E}\bigl[e^{tX_1}\bigr]
=p\,e^{t}+(1-p)\,e^{2t}
=e^t\bigl(p+(1-p)e^t\bigr),
\]
iterating the recursion yields the product formula.
\qed
\section{The entropy concavity conjecture}
The Shannon entropy of the random variables $X_{n}(p)$ is given by
\[ h_{n}(p)=H(X_{n}(p))=-\sum_{k=1}^{a_{n}}\mathbb{P}(X_{n}(p)=k)\log(\mathbb{P}(X_{n}(p)=k)),\]
where we use the convention that $0\log(0)=0$, see \cite{SCH}. We obviously have $h_{n}(0)=h_{n}(1)$. The symmetry \[h_{n}(p)= h_{n}(1-p)\] follows from \[\mathbb{P}(X_{n}(p)=k)=\mathbb{P}(X_{n}(1-p)=a_{n}+1-k).\] 
Now we state our entropy conjecture.
\begin{conjecture}
For every strictly increasing sequence of integers $(a_{n})$ with $a_{1}=2$ and all $n\ge 1$ the entropy function $h_{n}$ is concave on $[0,1]$ and hence unimodal with maximum at $1/2$.
\end{conjecture}
We prove this conjecture in two extremal cases. 
\begin{proposition} For $a_{n}=n+1$ and for all sequences satisfying $a_{n+1}\ge2a_n$ and $a_{1}=2$, Conjecture 3.1 holds.
\end{proposition}
\proof If $a_{n}=n+1$, as we already observed in the introduction $X_{n}(p)$ is just a shift of the Binomial distribution; $X_{n}(p)\sim \mbox{Binomial}(p,n)+1$. Such a Shift obviously does not change the entropy. Hence the result follows from the Shepp and Olkin theorem, that states that the Shannon entropy of the binomial distribution is concave, see \cite{SO}.\\
If $a_{n+1}\ge2a_n$ we obtain using the probability formula from the introduction 
\[  h_{n}(p)=-\sum_{k=1}^{2^n}p^{n-\sharp_{n}(k)} (1-p)^{\sharp_{n}(k)}\log(p^{n-\sharp_{n}(k)} (1-p)^{\sharp_{n}(k)})\]
\[ =-\sum_{i=0}^{n}\binom{n}{i}p^{n-i}(1-p)^{i}\log(p^{n-i}(1-p)^{i}) =-n(p\log(p)+(1-p)\log(1-p)),\]
where $\sharp_{n}(k)$ is the number of digits $1$ in the dyadic representation of $k-1$. 
Now it is straightforward that the function is in fact concave. \qed\\~\\
We have no proof of the concavity conjecture for other increasing infinite sequences. Nevertheless in the next section we obtain computational evidence for many classes of sequences. 
\section{Computational evidence} Given an increasing finite sequence $(a_{n})$ with $a_{1}=2$ we use a Mathematica program to test concavity of the entropy function $h_{n}(p)$ of the random variables $X_{n}(p)$. The Mathematica program used for all computations is available from the author. We found no counterexample among any of the tested sequences.
We were able to test linear sequence \[ (a_{n})=(2n),(3n-1),(5n-3),(10n-8)\] up to $n=200$. In figure 1 we display $h_{n}$ and $h_{n}^{\prime\prime}$ for $n=10,50,100$ in the case $(a_{n})=(2n)$. 
\begin{figure}
\vspace{0pt}\hspace{-25pt}\scalebox{0.3}{\includegraphics
{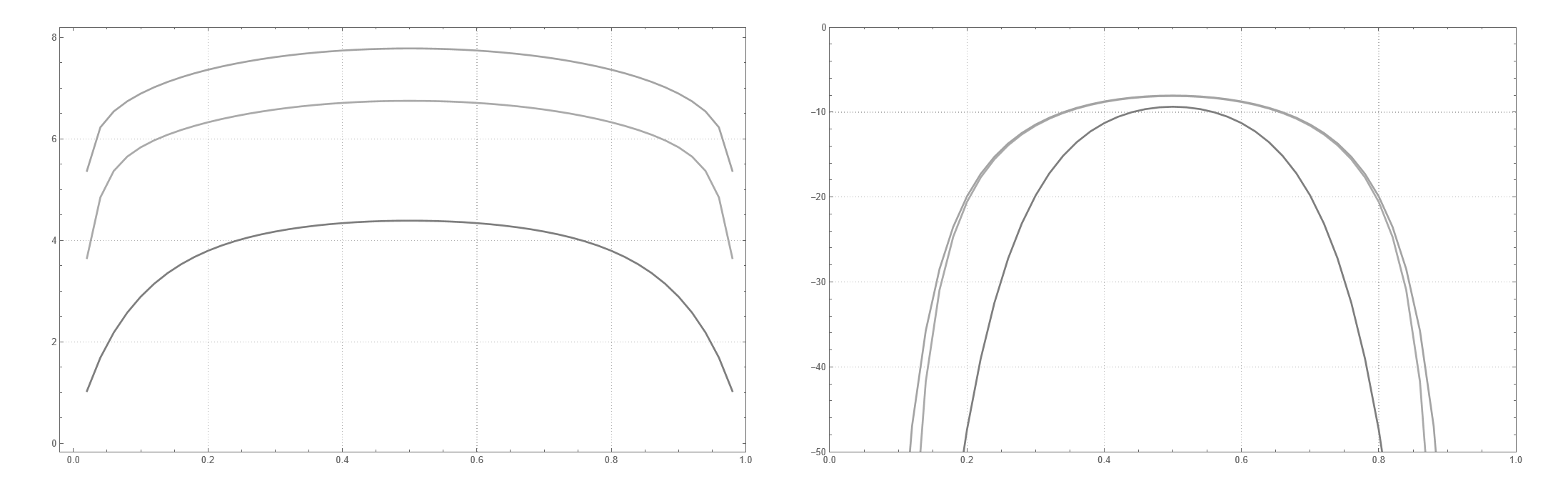}}
\caption{$h_{n}$ and $h_{n}^{\prime\prime}$ for $n=10,50,100$ in the case $(a_{n})=(2n)$. }
\end{figure}
We tested the polynomial sequences \[ (a_{n})=(n^2+1),(n^3+1)\] up to $n=100$ in the first case and $n=50$ in the second case. For the fist sequences We display $h_{n}$ and $h_{n}^{\prime\prime}$ for $n=10,50,100$ in figure 2.
\begin{figure}
\vspace{0pt}\hspace{-25pt}\scalebox{0.3}{\includegraphics
{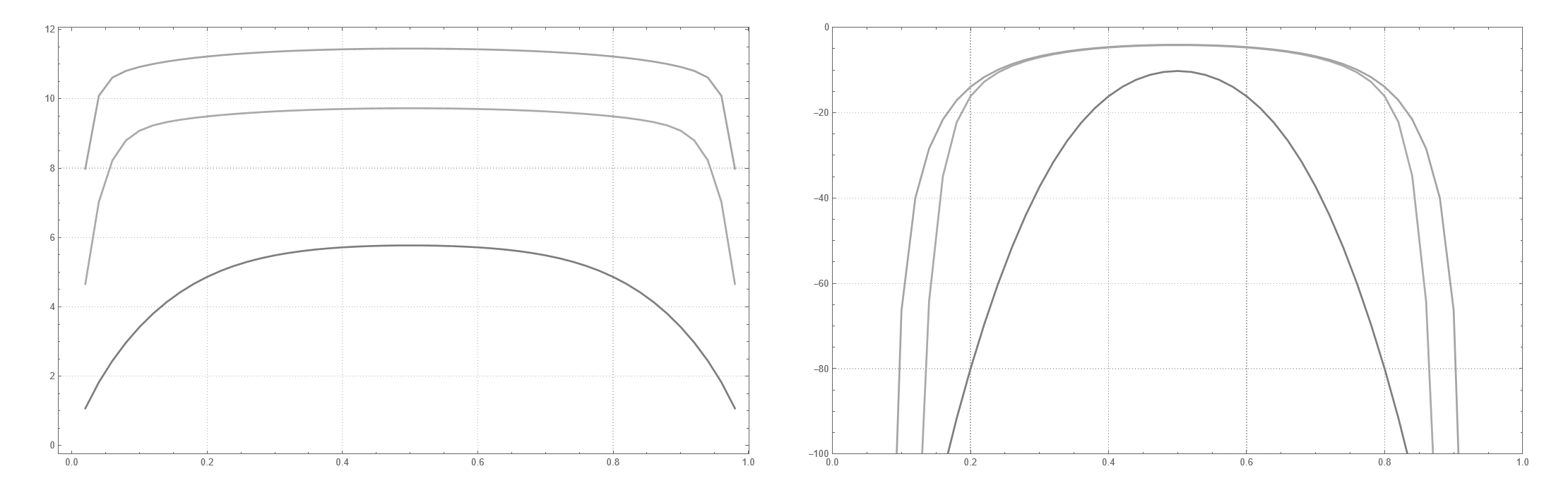}}
\caption{$h_{n}$ and $h_{n}^{\prime\prime}$ for $n=10,50,100$ in the case $(a_{n})=(n^2+1)$}
\end{figure}
We tested the exponential increasing sequences $a_{n}=\lceil(3/2)^{n}\rceil$ and the recursively definiere sequence $a_{n+1}=a_{n}+a_{n-1}+1$ with $a_{0}=1$ and $a_{0}=2$ up to $n=30$. In the second case we include a figure of $h_{n}$ and $h_{n}^{\prime\prime}$ with $n=10,15,20$ for the second sequence.
\begin{figure}
\vspace{0pt}\hspace{-25pt}\scalebox{0.3}{\includegraphics
{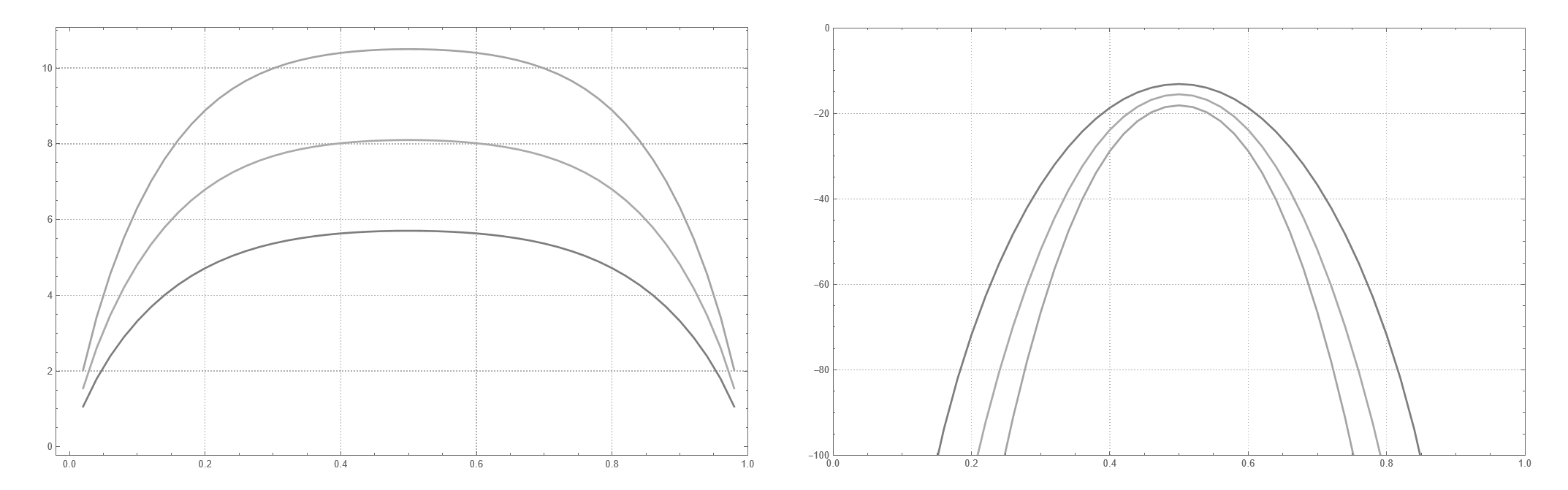}}
\caption{$h_{n}$ and $h_{n}^{\prime\prime}$ for $n=10,15,20$ in the case $a_{n+1}=a_{n}+a_{n-1}+1$}
\end{figure}
Moreover we testes stochastically increasing sequences up to $n=100$. The increment of the sequence was chosen uniformly from ${1,\dots,1000}$. The random experiments suggest that, if a counterexample exists, it must occur
only for considerably larger values of $n$ or for much more irregular overlap
sequences than those tested here.            

\end{document}